\pdfoutput=1 
\documentclass[11pt]{article}

\usepackage{lmodern}

\usepackage[margin=1in]{geometry}
\usepackage{amsfonts,amssymb,amsmath,amsthm,mathtools,thmtools}
\usepackage{microtype,xspace,booktabs}

\usepackage{graphicx}
\usepackage[bf,font=small]{caption}

\usepackage{caption}
\usepackage{subcaption}

\usepackage{hyperref}
\usepackage[usenames,dvipsnames]{xcolor}
\definecolor{newcolor}{hsb}{0.6,1,0.75}
\hypersetup{
	colorlinks=true,
	linkcolor=newcolor,
	citecolor=newcolor,
	filecolor=newcolor,
	urlcolor=newcolor
}

\usepackage[english]{babel}
\addto\extrasenglish{}
\addto\extrasenglish{}

\usepackage{enumitem}
\setlist{
  listparindent=\parindent, 
  parsep=0pt,
}

\usepackage{tikz}
\usetikzlibrary{arrows.meta,positioning} 

\usepackage[framemethod=tikz]{mdframed}
\mdfsetup{
	nobreak=true,
	roundcorner=5pt,
	linewidth=0.5,
	innerleftmargin=\parindent,
	innertopmargin=10pt,
	innerbottommargin=10pt,
}

\usepackage{lpic}

\makeatletter
\DeclareRobustCommand\bfseries{%
  \not@math@alphabet\bfseries\mathbf
  \fontseries\bfdefault\selectfont\boldmath}
\makeatother

\declaretheorem[name=Theorem]{theorem}
\declaretheorem[name=Lemma,sibling=theorem]{lemma}
\declaretheorem[name=Corollary,sibling=theorem]{corollary}
\declaretheorem[name=Claim,sibling=theorem]{claim}
\declaretheorem[name=Fact]{fact}
\declaretheorem[name=Remark,style=remark,sibling=theorem]{remark}

\declaretheorem[name=Puzzle]{puzzle}

\DeclareMathOperator{\C}{C}
\DeclareMathOperator{\UC}{UC}
\DeclareMathOperator{\s}{s}

\DeclareMathOperator{\bp}{bp}
\newcommand{\poly}{\textrm{poly}}
\newcommand{\adeg}{\widetilde{\vphantom{t}\smash{\deg}}}

\newcommand{\E}{\operatorname{\mathbb{E}}}
\renewcommand{\Pr}{\operatorname{\mathbb{P}}}
\newcommand{\Var}{\operatorname{Var}}

\newcommand{\Hex}{\textsc{\upshape\scshape Hex}}
\newcommand{\Eah}{\textsc{\upshape\scshape Eah}}
\newcommand{\CIS}{\textsc{\small\upshape\scshape CIS}}
\newcommand{\AND}{\textsc{\upshape\scshape And}}

\newcommand{\tO}{\tilde{O}}
\newcommand{\tOmega}{\tilde{\Omega}}

\newcommand{\calH}{\mathcal{H}}

\newcommand{\ol}[1]{\overline{#1}}

\begin{document}

\mbox{}\vspace{12mm}

\begin{center}
{\huge Unambiguous DNFs and Alon--Saks--Seymour}
\\[1cm] \large
	
\setlength\tabcolsep{2em}
\begin{tabular}{ccc}
Kaspars Balodis&
Shalev Ben-David&
Mika G\"o\"os\\[-1mm]
\small\slshape University of Latvia &
\small\slshape University of Waterloo &
\small\slshape EPFL
\end{tabular}

\vspace{2mm}
\begin{tabular}{cc}
Siddhartha Jain&
Robin Kothari\\[-1mm]
\small\slshape EPFL &
\small\slshape Microsoft Quantum
\end{tabular}
	
\vspace{6mm}
	
\large
\today
	
\vspace{6mm}
	
\end{center}

\begin{quote}
\leftskip4mm
\rightskip\leftskip
\noindent\small
{\bf Abstract.}~
We exhibit an unambiguous $k$-DNF formula that requires CNF width $\tOmega(k^2)$, which is optimal up to logarithmic factors. As a consequence, we get a near-optimal solution to the Alon--Saks--Seymour problem in graph theory (posed in 1991), which asks: How large a gap can there be between the chromatic number of a graph and its biclique partition number? Our result is also known to imply several other improved separations in query and communication complexity.
\end{quote}

\section{Three puzzles}

\subsection{First formulation}
An $n$-variate DNF formula $F= C_1\lor\cdots\lor C_m$ is said to be \emph{unambiguous} if for every input $x\in\{0,1\}^n$ at most one of the conjunctions $C_i$ evaluates to true, $C_i(x)=1$. If we think of the DNF formula as expressing its set of 1-inputs $F^{-1}(1)$ as a union of subcubes $C_i^{-1}(1)$, then $F$ is unambiguous precisely when the subcubes are pairwise disjoint. Unambiguity is a severe structural restriction on DNFs. In particular, every unambiguous DNF formula of bounded \emph{width} (defined as the maximum number of literals in a conjunction) can be written equivalently as a bounded-width CNF formula. Namely, we have the following folklore fact~\cite[\S III]{Goos2015}.

\begin{fact} \label{fact}
Every unambiguous $k$-DNF can be written equivalently as a $k^2$-CNF.	
\end{fact}

In this paper, we ask: Can this quadratic relationship be improved? Are there unambiguous $k$-DNFs that require CNFs of width much larger than $k$, perhaps even $\Omega(k^2)$? More formally, for a boolean function $f\colon\{0,1\}^n\to\{0,1\}$ we define the following standard complexity measures.
\begin{itemize}[noitemsep,label=$-$]
\item \emph{$1$-certificate complexity} $\C_1(f)$ is the least $k$ such that $f$ can be written as a $k$-DNF;
\item \emph{$0$-certificate complexity} $\C_0(f)$ is the least $k$ such that $f$ can be written as a $k$-CNF;
\item \emph{unambiguous $1$-certificate complexity} $\UC_1(f)$ is the least $k$ such that $f$ can be written as an unambiguous $k$-DNF.
\end{itemize}

\begin{puzzle}
\label{p1}
For $\alpha>1$, does there exist a boolean function $f$ with $\C_0(f) \geq \Omega(\UC_1(f)^\alpha)$?
\end{puzzle}

Here we abused terminology: instead of a single boolean function we really mean an infinite sequence of functions $f_n$ satisfying $\C_0(f_n)=\omega(1)$ as $n\to\infty$. \autoref{p1} was first asked in~\cite{Goos2015}, although an analogous question had been studied in communication complexity (under the name \emph{clique vs.\ independent set}; see \autoref{sec:result-apps}) since Yannakakis~\cite{Yannakakis1991}. The paper~\cite{Goos2015} gave a complicated recursive construction achieving an exponent $\alpha \approx 1.12$. This was subsequently optimised (but not simplified) in~\cite{BenDavid2017} improving the exponent to $\alpha\approx 1.22$.

\bigskip\noindent
{\bf Our main result} is a near-quadratic separation for \autoref{p1} (formally stated as \autoref{thm:gap} in \autoref{sec:results}), which matches the upper bound of \autoref{fact} up to logarithmic factors. Moreover, our construction is vastly simpler than previous ones.

\subsection{Second formulation}
In order to separate boolean function complexity measures it is often a good idea to proceed in two steps: First construct a partial boolean function $f\colon\{0,1\}^n\to\{0,1,*\}$ where some inputs $x$ are \emph{undefined}, $f(x)=*$. Then modify $f$ into a \emph{total} function by eliminating all the~$*$-inputs. We now formulate an appropriate partial function version of \autoref{p1}.

We recall the notion of a \emph{certificate}, adapted here for a partial function $f\colon\{0,1\}^n\to\{0,1,*\}$. Let $\Sigma\subseteq\{0,1,*\}$ be a subset of output symbols. We write for short $0,1,\ol{0},\ol{1}$ for the output sets $\{0\},\{1\},\{1,*\},\{0,*\}$. A partial input~$\rho\in \{0,1,*\}^n$ is a \emph{$\Sigma$-certificate} for $x\in\{0,1\}^n$ if $\rho$ is consistent with $x$ and for every input $x'$ consistent with~$\rho$ we have $f(x')\in \Sigma$. The \emph{size} of $\rho$, denoted $|\rho|$, is the number of its non-$*$ entries. The \emph{$\Sigma$-certificate complexity} of $x$, denoted $\C_\Sigma(f,x)$, is the least size of a $\Sigma$-certificate for $x$. The \emph{$\Sigma$-certificate complexity} of $f$, denoted $\C_\Sigma(f)$, is the maximum of~$\C_\Sigma(f,x)$ over all $x\in f^{-1}(\Sigma)$; this definition is consistent with the one given at the start of this section. Finally, we define \emph{certificate complexity} $\C(f)$ as $\max\{\C_0(f),\C_1(f)\}$.

\begin{puzzle}
\label{p2}
For $\alpha > 1$, does there exist a partial function $f$ together with an $x\in f^{-1}(*)$ such that both $\C_{\ol{0}}(f,x)$ and $\C_{\ol{1}}(f,x)$ are at least $\Omega(\C(f)^\alpha)$?
\end{puzzle}

We will show in \autoref{thm:equivalences} that \autoref{p1} and \ref{p2} are in fact equivalent: solving one with an exponent $\alpha$ will imply a solution to the other one with the same $\alpha$. The implication \ref{p1}$\Rightarrow$\ref{p2} is easy while the converse (converting a partial function into a total one) is non-trivial and uses the \emph{cheat sheet} framework introduced in~\cite{Aaronson2016}. Consequently, we feel that \ref{p2} is the most fruitful formulation to attack and that is indeed how our near-quadratic separation is obtained. 

\subsection{Third formulation}
We present one more equivalent formulation using purely graph theoretic language. While this version is not needed for our separation result, we include it for aesthetic reasons. Let $G=(V,E)$ be a hypergraph. We say $G$ is \emph{intersecting} if every two edges $e,e'\in E$ intersect, $e\cap e'\neq \emptyset$. A subset $U\subseteq V$ is a \emph{hitting set} for $G$ if $U$ intersects every edge $e\in E$. Moreover, $U$ is \emph{$c$-monochromatic} for a colouring $c\colon V\to\{0,1\}$ if $c$ is constant on $U$. Finally, we define the \emph{rank} of~$G$, denoted $r(G)$, as the maximum size $|e|$ of an edge $e\in E$. 

\begin{puzzle}
\label{p3}
For $\alpha > 1$, does there exist an intersecting hypergraph $G=(V,E)$ together with a colouring $c\colon V\to\{0,1\}$ such every $c$-monochromatic hitting set has size at least $\Omega(r(G)^\alpha)$?
\end{puzzle}

\autoref{p3} obscures the complexity-theoretic origins of the problem, thereby rendering it increasingly seductive for, say, an unsuspecting audience of combinatorialists (cf.~\cite{Raz2011}). In fact, we found all three formulations and proved them equivalent already in late 2015, and since then we have been deploying the camouflaged variant~\ref{p3} on several occasions, including, notably and most unsuccessfully, at an open problem seminar at the Institute for Advanced Study in 2018.

\section{Our contributions} \label{sec:results}

Our main results are as follows; here, the notation $\tOmega(n)$ hides $\poly(\log n)$ factors.
\begin{theorem} \label{thm:gap}
There exists a boolean function $f$ with $\C_0(f)\geq \tOmega(\UC_1(f)^2)$.
\end{theorem}
\begin{theorem} \label{thm:equivalences}
Puzzles \ref{p1}, \ref{p2}, \ref{p3} are near-equivalent: if one of them can be solved with exponent $\alpha$, then all of them can be solved with exponent $\alpha$ up to factors logarithmic in input length.
\end{theorem}

Our near-quadratic separation (\autoref{sec:quadratic}) is phrased as a solution to \autoref{p2} and so \autoref{thm:gap} follows from the equivalences in \autoref{thm:equivalences} (proved in \autoref{sec:equivalences}). We next discuss how our results imply several other separations in graph theory and query/communication complexity.

\begin{table}[!b]
\vspace{2mm}
\centering

\newcommand{\ub}{\color{ForestGreen}}
\newcommand{\lb}{\color{BrickRed}}

\renewcommand{\arraystretch}{1.1}
\setlength\tabcolsep{0.1em}
\begin{tabular}{l@{\hspace{1em}}l@{\hspace{4.5em}}l@{\hspace{2em}}l@{\hspace{2.5em}}l}
\toprule[.5mm]
\bf Reference &&& \bf $\chi(G)$ & \bf $\CIS_G$ \\
\midrule
Yannakakis & \cite{Yannakakis1991} & \ub $\forall G\colon$ & & $O(\log^2 n)$ \\
Mubayi and Vishwanathan & \cite{Mubayi2009} & \ub $\forall G\colon$ & $\exp(O(\log^2\bp(G)))$ \\
\midrule
Huang and Sudakov & \cite{Huang2012} & \lb $\exists G\colon$ & $\Omega(\bp(G)^{6/5})$ & $\geqslant\ 6/5 \cdot \log n$ \\
Amano & \cite{Amano2014} & \lb $\exists G\colon$ &  & $\geqslant\ 3/2 \cdot \log n$ \\
Shigeta and Amano & \cite{Shigeta2015} & \lb $\exists G\colon$ &  & $\geqslant\ 2 \cdot \log n$ \\
G\"o\"os & \cite{Goos2015} &\lb $\exists G\colon$ & $\exp(\Omega(\log^{1.12} \bp(G)))$ & $\Omega(\log^{1.12}n)$ \\
Ben-David et al. & \cite{BenDavid2017} & \lb $\exists G\colon$ & $\exp(\Omega(\log^{1.22} \bp(G)))$ & $\Omega(\log^{1.22} n)$ \\
\midrule
\bf This work & & \lb $\exists G\colon$ & $\exp(\tOmega(\log^2 \bp(G)))$ & $\tOmega(\log^2 n)$ \\
\bottomrule[.5mm]
\end{tabular}
\caption{{\ub\bf Upper} and {\lb\bf lower} bounds for the Alon--Saks--Seymour problem and for the conondeterministic communication complexity of clique vs.\ independent set problem. The two problems are near-equivalent~\cite{Bousquet2014}: a separation $\chi(G)\geq\bp(G)^c$ implies a conondeterministic lower bound $c\cdot \log n$ for some $\CIS_H$; conversely, a lower bound $c\cdot\log n$ for $\CIS_H$ implies a separation $\chi(G)\geq \Omega(\bp(G)^{c/2})$ for some $G$.}
\label{tab:bounds}
\end{table}

\subsection{Applications: Alon--Saks--Seymour and clique vs.\ independent set} \label{sec:result-apps}

The original motivation for studying \autoref{p1}~in \cite{Goos2015} was that its solutions imply lower bounds for two well-studied problems.
\begin{description}
\item[Alon--Saks--Seymour problem~{\upshape\cite{Kahn1991}.}]
For a graph $G$, how large can the chromatic number $\chi(G)$ be compared to the biclique partition number $\bp(G)$ (minimum number of complete bipartite graphs needed to partition the edges of $G$)?
\item[Clique vs.\ independent set problem~{\upshape\cite{Yannakakis1991}}.]
Define a two-party communication problem relative to an $n$-vertex graph $G=(V,E)$ as follows: Alice gets a clique $x\subseteq V$, Bob gets an independent set $y\subseteq V$, and their goal is to output $\CIS_G(x,y)\coloneqq |x\cap y|\in\{0,1\}$.
\end{description}
A surprising connection here is that constructing separations for the Alon--Saks--Seymour problem is equivalent to proving lower bounds on the \emph{conondeterministic} communication complexity of clique vs.\ independent set; see Bousquet et al.~\cite{Bousquet2014} for an excellent survey of this connection. \autoref{tab:bounds} summarises the progress on these two problems. In particular, Huang and Sudakov~\cite{Huang2012} were the first to find a polynomial separation between $\chi(G)$ and $\bp(G)$, which disproved a conjectured linear relationship due to Alon, Saks, and Seymour (formulated in 1991~\cite{Kahn1991}). Subsequent work has found alternative polynomial separations~\cite{Cioaba2011,Amano2014,Shigeta2015}, including improved lower bounds for $\CIS_G$. The first superpolynomial separation between $\chi(G)$ and $\bp(G)$ was obtained in~\cite{Goos2015}. This was achieved by employing a \emph{lifting theorem}~\cite{Goos2016} that converts a solution to \autoref{p1} with exponent $\alpha$ into a graph $G$ witnessing a separation $\chi(G)\geq \exp(\tOmega(\log^{\alpha}\bp(G)))$, or equivalently, into a conondeterministic lower bound $\tOmega(\log^\alpha n)$ for some $\CIS_H$. If we plug~\autoref{thm:gap} into the lifting framework of~\cite{Goos2015,Goos2016} we get near-optimal lower bounds for the two problems.%
\begin{corollary} \label{cor:cis}
There exists a graph $G$ such that $\chi(G)\geq \exp({\tOmega(\log^2\bp(G))})$. Equivalently, there exists a graph $H$ such that $\CIS_H$ requires $\tOmega(\log^2n)$ bits of conondeterministic communication. 
\end{corollary}

Let us pause here to appreciate how long a chain of reductions we have created to solve a graph theoretic problem by a reduction to another (hyper)graph theoretic problem, but fundamentally passing through complexity theory. That is, we have
\begin{enumerate}[noitemsep]
\item \label{i} Alon--Saks--Seymour problem, reduces to
\item \label{ii} clique vs.\ independent set problem, reduces to
\item \label{iii} \autoref{p1}: separation $\C_0 \gg \UC_1$ in query complexity, reduces to
\item \label{iv} \autoref{p2}: separation $\C_{\ol{0}},\C_{\ol{1}} \gg \C$ for a partial function, reduces to
\item \label{v} \autoref{p3}: $2$-colouring an intersecting hypergraph.
\end{enumerate}
Reduction 2-to-3 uses a lifting theorem (which is not known to have a converse) and 3-to-4 uses cheat sheets---both of these are inherently query/communication tools that do not have natural counterparts in classical combinatorics. The end result can be phrased as its own graph problem: Given an intersecting hypergraph and a 2-colouring whose monochromatic hitting sets are power-$\alpha$ larger than the rank, construct a graph which is an edge-disjoint union of $k$ bicliques but which has chromatic number $\exp(\tOmega(\log^{\alpha} k))$. It sounds to us magical that this is possible!

\paragraph{Other related work.}
Previously, near-optimal $\tOmega(\log^2 n)$ communication lower bounds for $\CIS_G$ were known in the deterministic~\cite{Goos2018a} and even randomised~\cite{Goos2018b} communication models. However, these results do not imply any bounds for the conondeterministic complexity and hence neither for the Alon--Saks--Seymour problem.

Given that superpolynomial separations exist for the Alon--Saks--Seymour problem in general, a recent line of work has aimed to find special graph classes where the separation is at most polynomial~\cite{Bousquet2014,Lagoutte2016,Bousquet2018,Chudnovsky2021}. In particular, it remains open whether the separation is polynomial for the class of perfect graphs.

\subsection{Applications: Separations in query complexity}\label{sec:querysep}

In query complexity, we get three improved separations involving the well-studied complexity measures \emph{degree} $\deg(f)$, \emph{sensitivity} $\s(f)$, and \emph{approximate degree} $\adeg(f)$ (defined in \autoref{sec:apps}). We refer to Aaronson et al.~\cite{Aaronson2021} for an up-to-date survey of the known relationships.

\begin{corollary} \label{cor:deg}
There exists a boolean function $f$ with $\C(f) \geq \tOmega(\deg(f)^2)$.
\end{corollary}
\begin{corollary} \label{cor:sens}
There exists a boolean function $f$ with $\C(f) \geq \tOmega(\s(f)^3)$.
\end{corollary}
\begin{corollary} \label{cor:adeg}
There exists a boolean function $f$ with $\C(f) \geq \tOmega(\adeg(f)^3)$.
\end{corollary}

\autoref{cor:deg} follows from \autoref{thm:gap} and the simple fact that $\UC_1(f)\geq \deg(f)$; in particular, our quadratic separation improves over the classic power-$1.63$ separation due to Nisan, Kushilevitz, and Wigderson~\cite{Nisan1995}. \autoref{cor:sens} follows automatically from~\cite[Theorem 1]{BenDavid2017}; that work exhibited a power-$2.22$ separation, which was the previous record.

\autoref{cor:adeg} is the trickiest. We do not know how to derive it from our quadratic solution to \autoref{p1}. Instead, we will present (\autoref{sec:hex}) an alternative solution with $\alpha=1.5$ (already beating $\alpha\approx 1.22$ from prior work), which is even simpler and more structured than our quadratic one and hence more useful in deriving the cubic gap for \autoref{cor:adeg}. The previous best separation here was quadratic as witnessed by the $n$-bit $\AND$ function.

\section{Quadratic solution to \autoref{p2}} \label{sec:quadratic}

In this section, we will describe our near-quadratic solution to \autoref{p2}. Namely, we will construct a partial boolean function $\Eah_n\colon\{0,1\}^{O(n^2)}\to\{0,1,*\}$ and an input $z\in \Eah_n^{-1}(*)$ such that
\begin{align}
\C(\Eah_n) ~\coloneqq~ \max\big\{\C_0(\Eah_n),\,\C_1(\Eah_n)\,\big\} ~&\leq~ \tO(n), \label{eq:eah-ub}\\
\min\big\{\C_{\ol{0}}(\Eah_n,z),\,\C_{\ol{1}}(\Eah_n,z)\,\big\} ~&\geq~ \Omega(n^2). \label{eq:eah-lb}
\end{align}
Our construction centers around a hypergraph with a certain \emph{pseudorandom} property as formalised in \autoref{lem:eah} below. We first use this lemma in \autoref{sec:construction} to construct the function $\Eah_n$ and then we prove the lemma in \autoref{sec:eah-graph}.

\begin{lemma}[EAH graphs] \label{lem:eah}
There exists an $n$-uniform hypergraph $G=(V,E)$ with $|V|=n^2$ vertices and $|E|=n^2$ edges that satisfies the following \emph{``everywhere almost-hittable'' (EAH)} property: For every set $F\subseteq V$ of size $|F|\leq \frac{1}{100}n^2$, there exists a set $H\subseteq V$ such that
\begin{itemize}[noitemsep,label=$-$]
\item $H$ has size at most $\tilde{n}\coloneqq 100n\log n$, 
\item $H$ is disjoint from $F$,
\item $H$ intersects all but at most $n$ of the edges in $E$.
\end{itemize}
(In other words, for every set of ``forbidden'' vertices $F$, there is a small hitting set $H$ that does not use the forbidden vertices and that hits almost all of the edges.)
\end{lemma}

\subsection{Quadratic separation from an EAH graph.} \label{sec:construction}
Fix an $n$-uniform EAH hypergraph $G=(V,E)$ given by \autoref{lem:eah}. We define a partial function $\Eah_n\colon\{0,1\}^{V\cup E}\to\{0,1,*\}$ that has an input variable $x_v$ for each vertex $v$ and an input variable~$x_e$ for each edge $e$. We set $f(x)\coloneqq 1$ iff there is some edge $e$ such that $x_e=1$ and $x_v=1$ for all $v\in e$. We set $\Eah_n(x)\coloneqq 0$ iff $x$ is not a 1-input and there is a certificate of this fact that uses at most $2\tilde{n}$ bits. If neither of these cases hold, we set $\Eah_n(x)\coloneqq *$. In summary,
\[
\Eah_n(x) ~\coloneqq~
\begin{cases}
	~1 & \text{if there is an edge $e$ such that $x_e=1$ and $x_v=1$ for all $v\in e$},\\
	~0 & \text{if } \C_{\ol{1}}(\Eah_n,x)\leq 2\tilde{n},\\
	~* & \text{otherwise}.
\end{cases}
\]

By construction, $\C_1(\Eah_n) = n+1$ and $\C_0(\Eah_n) \leq 2\tilde{n}$, which verifies \eqref{eq:eah-ub}. It remains to find a~$*$-input $z$ satisfying \eqref{eq:eah-lb}. Consider an input $z$ where $z_v=1$ for all vertices $v$ and $z_e=0$ for all edges~$e$. Clearly $\Eah_n(z)\neq 1$. Moreover, $\C_{\ol{1}}(\Eah_n,z)=n^2$ since any $\ol{1}$-certificate needs to read all the edge-variables $z_e$. Hence $\Eah_n(z)=*$. The following claim verifies \eqref{eq:eah-lb} and completes the proof.

\begin{claim}
$\C_{\ol{0}}(\Eah_n,z) \geq \Omega(n^2)$.
\end{claim}
\begin{proof}
Let $\rho$ be a partial input consistent with $z$ that has size $\frac{1}{100}n^2$. We show that $\rho$ cannot be a $\ol{0}$-certificate. Denote by $F\subseteq V$, $|F|\leq \frac{1}{100}n^2$, the set of vertices read by $\rho$. By the EAH property, there is a set $H\subseteq V\setminus F$, $|H|\leq \tilde{n}$, that hits all edges in $E\setminus M$ for some~$M\subseteq E$, $|M|\leq n$. Consider flipping the bits of $z$ corresponding to vertices $H$ to $0$. This gives us a string $w$ that is still consistent with $\rho$. However, we claim that $\Eah_n(w)=0$ (which shows that $\rho$ is not a $\ol{0}$-certificate, completing the proof). The reason is that we can read the vertex-variables in $H$ (which are all $0$ in $w$) together with the edge-variables in $M$ and this forms a~$\ol{1}$-certificate for $w$ of size at most $\tilde{n}+n\leq 2\tilde{n}$.
\end{proof}

\subsection{Existence of EAH graphs (Proof of \autoref{lem:eah})} \label{sec:eah-graph}

It is not hard to prove that a random $n$-uniform hypergraph with the required number of vertices/edges satisfies the conditions of \autoref{lem:eah}. However, we give here an explicit construction.

Define $V\coloneqq [n]\times[n]$ and think of these vertices as being arranged in an $n$-by-$n$ grid. We will consider size-$n$ edges that will contain exactly one vertex from each row of the grid. Thus the edges are in $1$-to-$1$ correspondence with functions $h\colon [n]\to[n]$ where a function $h$ corresponds to the edge $e_h\coloneqq \{(i,h(i)) : i\in [n]\}$. Let $\calH$, $|\calH|=n^2$, be any family of \emph{pairwise independent hash functions}, that is, satisfying (we write $h\sim\calH$ for a uniform random $h\in\calH$)
\begin{align}
\forall i,j\colon\quad&
\Pr_{h\sim\calH}[h(i)=j] ~=~1/n, \label{eq:unif} \\
\forall i\neq i',j,j'\colon\quad&
\Pr_{h\sim\calH}[h(i)=j \text{ and } h(i')=j'] ~=~ 1/n^2. \label{eq:pairwise}
\end{align}
(For the most basic example, assume $n$ is a prime power and identify $[n]$ with the field $\mathbb{F}_n$. Define a function family $\calH=\{h_{a,b}\}$ indexed by $a,b \in \mathbb{F}_n$ such that $h_{a,b}(i) \coloneqq ai+b$.) We let $E\coloneqq\{e_h:h\in \calH\}$. This completes the construction of $G\coloneqq(V,E)$.

To verify the EAH property, fix a forbidden set of vertices $F \subseteq [n]\times[n]$ of size $|F|\leq \frac{1}{100}n^2$. Define a set of \emph{mostly-forbidden} edges $M \coloneqq \{e\in E : |e\cap F| \geq \frac{9}{10}n \}$. The following two claims finish the proof of \autoref{lem:eah}.
\begin{claim}
$|M|\leq n$.
\end{claim}
\begin{proof}
By averaging, for at least half the rows, $F$ contains at most $pn$, $p\coloneqq 1/50$, vertices from each of those rows. Suppose for simplicity that this happens for the first $n'\coloneqq n/2$ rows, and suppose further that $F$ contains \emph{exactly} $pn$ vertices from each such row (which can be ensured by adding more vertices to $F$).  Let $h\sim\calH$ henceforth. Define for $i \in [n']$ an indicator random variable $X_i\in\{0,1\}$ such that $X_i=1$ iff $(i,h(i)) \in F$. Then \eqref{eq:unif} implies that $\Pr[X_i=1] = p$ and~\eqref{eq:pairwise} implies that the $X_i$ are pairwise independent. Consider $X \coloneqq \sum_{i=1}^{n'} X_i$. This has expectation $\E[X]=pn'$ and variance $\Var[X] = n'p(1-p)$ like the $p$-biased binomial distribution. We calculate
\begin{align*}\textstyle
\Pr[e_h \in M]
~\leq~ \Pr[ X \geq \frac{4}{5}n']
~\leq~ \Pr[ X - \E[X] \geq \frac{1}{2}n' ]
~\leq~ \Var[X]/(\frac{1}{2}n')^2
~\leq~  4p/n'
~\leq~ 1/n,
\end{align*}
where we used Chebyshev's inequality. We conclude that $|M|=n^2\cdot\Pr[e_h\in M]\leq n$.
\end{proof}

\begin{claim}
There exists a set $H\subseteq V\setminus F$ of size $\tilde{n}$ that intersects every edge in $E\setminus M$.
\end{claim}
\begin{proof}
We claim that a uniform random $\tilde{n}$-subset $H\subseteq V\setminus F$ satisfies the claim with high probability. Consider a fixed $e\in E\setminus M$ so that $|e\setminus F| > n/10$. A birthday-paradox-like calculation gives
\[\textstyle
\Pr_H[e\cap H = \emptyset]
~\leq~ (1-\frac{1}{10n})^{\tilde{n}}
~=~ [(1-\frac{1}{10n})^{10n}]^{10\log n}
~=~ [1/e - o(1)]^{10\log n}
~\ll~ 1/n^2.
\]
A union bound over all $e\in E\setminus M$ shows that $H$ hits every edge in $E\setminus M$ with high probability.
\end{proof}

\section{Power-1.5 solution to \autoref{p2}}  \label{sec:hex}

In this section, we describe an alternative solution to \autoref{p2} with exponent $\alpha=1.5$ (which also beats the previous best construction with $\alpha\approx 1.22$~\cite{BenDavid2017}). Our alternative solution is more structured than the quadratic one, which allows us to use it to prove \autoref{cor:adeg} in \autoref{sec:apps}. Our construction is inspired by the board game Hex.

We define a partial boolean function $\Hex_n\colon\{0,1\}^{n\times n}\to\{0,1,*\}$ whose $n^2$-bit inputs are interpreted as $n\times n$ boolean matrices. We say that two matrix entries in~$[n]\times[n]$ are \emph{connected} if they are adjacent either horizontally or vertically (but not diagonally). A~\emph{$1$-path} in an input $x$ is top-to-bottom path of $1$-entries, that is, the path starts on a $1$-entry in the topmost row, moves along connected $1$-entries, and ends on the bottommost row. Similarly, a \emph{$0$-path} in $x$ is a left-to-right path of $0$-entries. Note that no $x$ can contain both a $1$-path and a $0$-path. We define
\[
\Hex_n(x) ~\coloneqq~
\begin{cases}
	~1 & \text{if $x$ contains a $1$-path of length at most $2n$},\\
	~0 & \text{if $x$ contains a $0$-path of length at most $2n$},\\
	~* & \text{otherwise}.
\end{cases}
\]
Clearly $\C(\Hex_n) = 2n$. It remains to prove the following lemma. For simplicity, we drop $\Hex_n$ from notation and write $\C_\Sigma(x)\coloneqq \C_\Sigma(\Hex_n,x)$.
\begin{lemma} \label{lem:hard}
There is an $x\in\Hex_n^{-1}(*)$ such that both $\C_{\ol{0}}(x)$ and $\C_{\ol{1}}(x)$ are $\Omega(n^{1.5})$.
\end{lemma}

\begin{figure}[t]
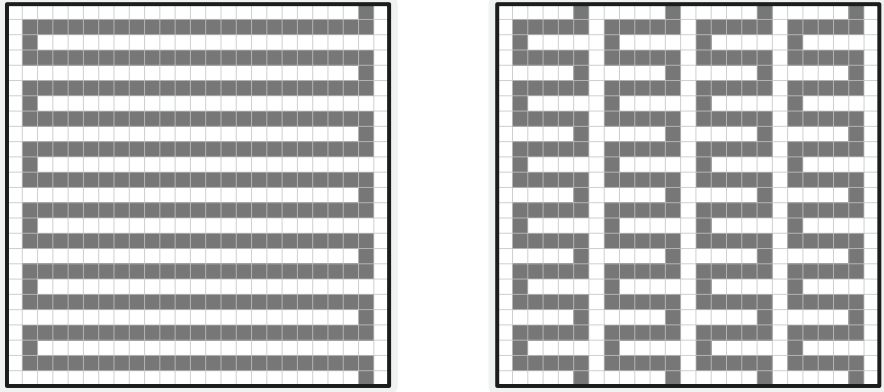

\centering
\begin{lpic}[t(-6mm)]{figs/inputs(.25)}
\lbl[c]{120,7;(a) input $y$}
\lbl[c]{390,7;(b) input $z$}
\end{lpic}
\caption{Inputs to $\Hex_n$ are $n\times n$ boolean matrices. Illustrated are two $*$-inputs: $y$ consists of a single $1$-spiral of length $\Theta(n^2)$, and  $z$ consists of $\sqrt{n}$ many $1$-spirals of length $\Theta(n^{1.5})$ each.}
\label{fig:inputs}
\vspace{-1mm}
\end{figure}

We note that it is easy to find inputs $x$ where one of $\C_{\ol{0}}(x)$ or $\C_{\ol{1}}(x)$ is large, but not both. For example, consider the input $y$ depicted in \hyperref[fig:inputs]{\autoref{fig:inputs}a} that contains a single spiralling $1$-path, call it a \emph{$1$-spiral} for short. The $1$-spiral has length $\Theta(n^2)> 2n$ and hence $\Hex_n(y)=*$.
\begin{claim}
$\C_{\ol{0}}(y)\geq\Omega(n^2)$ and $\C_{\ol{1}}(y) \leq O(n)$.
\end{claim}
\begin{proof}
For the first claim, we employ a sensitivity argument. Consider any entry $e\in[n]\times[n]$ in the $1$-spiral that is not a \emph{corner} (where the spiral makes a right-angle turn). Denote by $y^e$ the input $y$ but with the entry $e$ flipped (from $1$ to $0$). Note that flipping $e$ introduces a short ($\leq 2n$) $0$-path in~$y^e$ and thus $\Hex_n(y^e)=0$. It follows that any $\ol{0}$-certificate for $y$ needs to read all the non-corner entries of which there are $\Theta(n^2)$ many. For the second claim, we note that it suffices to include the five topmost rows in a certificate to prove that any $1$-path in $y$ must be of length $>2n$.
\end{proof}

We can similarly find an input $y^*$ with large $\C_{\ol{1}}(y^*)$ and small $\C_{\ol{0}}(y^*)$. The key challenge is to find a single input where both $\ol{0}$- and $\ol{1}$-complexities are large. Our solution is to ``balance'' $y$. Namely, we let $z$ be the input that consists of $\sqrt{n}$ many disjoint $1$-spirals, each of length $\Theta(n^{1.5})$; see \hyperref[fig:inputs]{\autoref{fig:inputs}b} for an illustration. The following two claims complete the proof of \autoref{lem:hard}.

\begin{claim}
$\C_{\ol{0}}(z)\geq\Omega(n^{1.5})$.
\end{claim}
\begin{proof}
\definecolor{myblue}{HTML}{82b3ff}
We employ a block sensitivity argument. Let $\ell = \Theta(n^{1.5})$ denote the number of non-corner entries in each $1$-spiral of $z$. For each $i\in[\ell]$, we define a \emph{block} $B_i\subseteq [n]\times[n]$, $|B_i|=\sqrt{n}$, as the set that contains the $i$-th non-corner entry from each $1$-spiral (the non-corners of a spiral are ordered top-to-bottom, say). Denote by~$z^{B_i}$ the input obtained from $z$ by flipping all the entries in $B_i$ (from $1$ to $0$). Flipping any block $B_i$ introduces a short ($\leq 2n$) $0$-path in $z^{B_i}$ and hence $\Hex_n(z^{B_i})=0$. For example, in the following illustration, the short $0$-path (drawn in {\bf\color{myblue} blue}) is created when we flip the block consisting of the striped entries:
\begin{center}
\begin{lpic}[t(-4mm),b(-5mm)]{figs/flip(.18)}
\end{lpic}
\end{center}
It follows that any $\ol{0}$-certificate for $z$ needs to read at least one entry from each of the blocks. But since the blocks are disjoint and there are $\ell=\Theta(n^{1.5})$ many of them, the claim is proved.
\end{proof}

\begin{claim}
$\C_{\ol{1}}(z)\geq\Omega(n^{1.5})$.
\end{claim}
\begin{proof}
Let $\rho$ be a partial input consistent with $z$ that has size $o(n^{1.5})$. We show that $\rho$ cannot be a $\ol{1}$-certificate. By averaging, there is some ``neglected'' $1$-spiral such that $\rho$ reads $o(n)$ many $0$-entries adjacent to the spiral. We will greedily construct a $1$-path  consistent with $\rho$ by starting at the top of the neglected spiral and trying to fit a $1$-path straight down the matrix. The $0$-entries read by $\rho$ can prevent a direct downward path from working, but every time we encounter such a $0$-entry we can avoid it by taking one step to the left or right (following the direction of the spiral). These left/right steps make us waste at most $o(n)$ extra steps in addition to the $n$ downward steps. This shows there exists a $1$-path of length $n+o(n)\leq 2n$ consistent with $\rho$, and hence $\rho$ is not a $\ol{1}$-certificate.
\end{proof}

\begin{remark}
It is easy to see that $\C_{\ol{0}}(z)$ and $\C_{\ol{1}}(z)$ are also $O(n^{1.5})$. We suspect that $z$ is in fact extremal for $\Hex_n$ meaning that no other $*$-input can witness an exponent larger than $\alpha=1.5$. However, we have not been able to prove this. Any improvement over $\alpha=1.5$ would translate into a better separation in \autoref{cor:adeg}.
\end{remark}

\section{Equivalences of puzzles} \label{sec:equivalences}

We now prove our three puzzles equivalent (\autoref{thm:equivalences}). The proof comprises of four implications,  each proved in its own subsection: \ref{p2} $\Rightarrow$ \ref{p1} $\Rightarrow$ \ref{p2} $\Rightarrow$ \ref{p3} $\Rightarrow$ \ref{p2}. This is more implications than strictly necessary, but not all directions are equally good in terms of overheads caused by log factors.


\subsection{Construction \ref{p2} \texorpdfstring{$\Rightarrow$}{=>} \ref{p1}} \label{thm:twotoone}
\vspace{10pt}
\begin{mdframed}
\begin{tabular}{rl}
\bf Given: & A partial function $f\colon \{0,1\}^n \to \{0,1,*\}$ and an input $x\in f^{-1}(*)$.\\[6pt]
\bf Construct: & A total function $g\colon\{0,1\}^{3n^2\log^2 n}\to\{0,1\}$ such that \\
& $\C_0(g) \ge \min\{\C_{\ol{0}}(f,x),\C_{\ol{1}}(f,x)\}$ and $\UC_1(g) \le 3\C(f) \log^2 n$.
\end{tabular}
\end{mdframed}

\paragraph{Overview.}
The basic idea is that we would like to turn regular,
ambiguous certificates for $f$ into unambiguous collections of certificates
for a modified function $g$. One way to do so is to give each
certificate for $f$ a unique identification number; then we can
require the new inputs to $g$ to consist of both an input $z$ to $f$
and an identification number (written in binary) for a certificate
in $z$. We will let such a new input $(z,k)$ evaluate to $1$
if the certificate specified by the number $k$ really is in $z$,
and we will define $(z,k)$ to be a $0$-input otherwise.
Then by reading all of $k$ and the corresponding
certificate in $z$, we get an unambiguous certificate for $(z,k)$
whenever $(z,k)$ is a $1$-input.

This strategy makes $1$-certificates unambiguous, but it
does not necessarily ensure that the function is hard to certify
on $0$-inputs. The reason is that for the new function,
it is conceivable that we could certify $(z,k)$ is a $0$-inputs
just by reading a few bits of $k$ and a few bits of $z$, but
that those few bits suffice to prove that the certificate specified
by $k$ cannot possibly be found in $z$. Indeed, it might even
be easy to certify that $(z,k)$ is a $0$-input when $z=x$,
the hard $*$-input to $f$.

We wish to eliminate this $0$-certification strategy so that
the new function is hard to certify on at least one $0$-input. To do so,
we will use the \emph{cheat sheet} framework~\cite{Aaronson2016}.
The idea is to hide the identification number $k$ of the
certificate in one cell of an array consisting of, say, $n$ different 
cells. We choose $n$ cells because this is large enough so that even 
reading a single bit from each cell is too expensive. 
But now that we have hidden $k$ in one of $n$ cells, there needs to
be a way to find it.
So to specify which cell of the array is the ``correct'' one, 
the one where we've stored $k$,  
we will change the problem to have $\log n$ different instances of $f$,
and we will interpret the $f$-outputs of these instances as specifying
a binary string of length $\log n$, which can index a specific cell of our array. 
Now that there are $\log n$ copies of $f$, the correct array cell will 
be required to contain identification numbers for $\log n$ different certificates,
one for each instance of $f$. 
Now we define this new function $g$ to evaluate to $1$ if 
all the $\log n$ $f$-inputs are $0$- or $1$-inputs
and if the array cell indexed really contains valid identification
numbers of certificates present in the $f$-inputs. 
Otherwise, if this doesn't hold, we define the input to be a $0$-input to $g$. 

With this construction, the contents of the cell pointed to by the $\log n$-bit 
string of outputs of $f$, along with the certificates in that 
cell form small unambiguous certificates for $1$-inputs to $g$.
On the other hand, the $g$-input consisting of $\log n$
copies of $x \in f^{-1}(*)$ together with any array content will be 
a $0$-input that is hard to certify: 
Any certificate must either prove that
at least one copy of $x$ is not a $0$-input or not a $1$-input,
which is expensive to do because we assumed that 
$\C_{\ol{0}}(f,x)$ and $\C_{\ol{1}}(f,x)$ are large, 
or else it must prove that none of the cells in the array contain valid certificates, which requires it to read at least one bit from each of the $n$ cells. We now prove this more formally.

\paragraph{Formal proof.}
A certificate of size $\C(f)$ specifies the indices of up to $\C(f)$ input bits and an assignment to those bits. Since an index can be encoded using $\log n$ bits, the total number of bits needed to represent a certificate is at most $\ell \coloneqq 2\C(f)\log n \leq 2n \log n$. We choose~$k \coloneqq \log n$ as the number of copies of $f$ that we will use.
Let us define $g\colon\{0,1\}^{kn+2^k k\ell}\to\{0,1\}$ on 
$kn+2^k k \ell \leq 3n^2\log^2 n$ bits.
For an input $z$ to $g$, we define $s_z$ to be the string in
$\{0,1,*\}^k$ that we get by applying $f$ to the first $kn$
bits of $z$, interpreted as $k$ inputs to $f$.
If $s_z\notin\{0,1\}^k$, we define $g(z)\coloneqq 0$.
Otherwise, if $s_z\in \{0,1\}^k$, we interpret the last~$2^k k \ell$ bits of $z$ as an array of $2^k$ cells of size $k\ell$
each, and we let $C_z\in \{0,1\}^{k\ell}$ be the contents
of the cell indexed by $s_z$. We interpret~$C_z$ as specifying
$k$ different certificates for $f$, each specified using $\ell$ bits.
We then set $g(z)\coloneqq 1$ if each of the $k$ inputs for $g$
in the first part of the string~$z$ contains the corresponding certificates 
specified by $C_z$ in order. Otherwise, we set $g(z)\coloneqq 0$.

The following two claims verify that this construction has the desired properties.

\begin{claim}
$\C_0(g)\ge \min\{\C_{\ol{0}}(f,x),\C_{\ol{1}}(f,x)\}$.
\end{claim}
\begin{proof}
Consider the input to $g$ consisting of $k$
copies of $x$, followed by an all-$0$ array. Consider any
certificate $c$ for this input. If $c$ reads fewer than
$\min\{\C_{\ol{0}}(f,x),\C_{\ol{1}}(f,x)\}$ bits, then $c$
does not certify that $x$ is not a $0$-input or that $x$ is not
a $1$-input for any of the $k$ copies of $x$. Moreover,
$c$ also cannot read a bit of each of the $2^k = n$ array cells,
since $n$ is larger than $\min\{\C_{\ol{0}}(f,x),\C_{\ol{1}}(f,x)\}$.
Hence there is some array cell, indexed by some string $s\in\{0,1\}^k$,
such that $c$ reads no bits of that array cell. Since $c$
fails to prove anything about the $f$-outputs of the copies of $x$,
we can find an input $y$ to $g$ which is consistent with $c$
such that the $f$-inputs in $y$ evaluate to $s$; moreover,
we can then set the array cell indexed by $s$ to provide
valid certificates for the $k$ inputs to $f$ in $y$. This causes
$y$ to be a $1$-input consistent with $c$, contradicting
the assumption that $c$ was a $0$-certificate.
\end{proof}

\begin{claim}
$\UC_1(g) \leq 3 \C(f) \log^2 n$.
\end{claim}
\begin{proof}
Intuitively, the contents of the cell 
referred to by the string $s_z$ and all the certificates in it 
together form an unambiguous certificate for $f$.
So an unambiguous $1$-certificate for $g$ has the following form: 
first, it reads exactly one certificate for each of the $k$ inputs to $f$;
second, in the array cell indexed by $s_z$, the $\log n$-bit string 
of $f$-outputs, the certificate reads the entire array cell, 
and the cell has the property that it contains exactly the same 
certificates read in the $k$ inputs to $f$ (in order). 
The size of this certificate is $k(\ell+\C(f)) \leq \log n(2 \C(f)\log n + \C(f))\leq 3 \C(f) \log^2 n$ where $k\C(f)$ bits are used to specify the certificates for $k$ copies of $f$ and $k\ell$ bits are used to read the full contents of one cell of the array.

It remains to show that the above collection of 1-certificates is unambiguous.
We claim that no input $z$ to $g$ can have two different certificates of the type we have just described.
To see this, suppose otherwise, and let $\rho_1$ and $\rho_2$
be two such certificates consistent with $z$.
Suppose that $\rho_1$ reads bits in the array cell $C_1$
and that $\rho_2$ reads bits in the array cell $C_2$.
Then since $\rho_1$ proves that the $f$-inputs in $z$
evaluate to the index of $C_1$, and since $\rho_2$ proves
that the $f$-inputs in $z$ evaluate to the index of $C_2$,
we must have $C_1=C_2$. Since $\rho_1$ reads all of $C_1$
and $\rho_2$ reads all of $C_2$, we know that $\rho_1$ and $\rho_2$
are identical on the array part of the input.
However, this array cell then specifies exactly
which bits a certificate in this collection must read from the $k$
inputs to $f$; it follows that $\rho_1$ and $\rho_2$ must be identical.
\end{proof}

\subsection{Construction \ref{p1} \texorpdfstring{$\Rightarrow$}{=>} \ref{p2}}
\vspace{8pt}
\begin{mdframed}
\begin{tabular}{rl}
\bf Given: & A total function $f\colon \{0,1\}^n \to \{0,1\}$.\\[6pt]
\bf Construct: & A partial function $g\colon\{0,1\}^{2n}\to\{0,1,*\}$ and an input
$z\in g^{-1}(*)$\\
& such that $\min\{\C_{\ol{0}}(g,z),\C_{\ol{1}}(g,z)\}\ge \C_0(f)$ and $\C(g)\le 2\UC_1(f)$.
\end{tabular}
\end{mdframed}

\noindent
Let $U\subseteq\{0,1,*\}^n$ be an unambiguous collection of $1$-certificates for $f$ so that
\begin{itemize}[noitemsep,label=$-$]
\item for every $x\in f^{-1}(1)$ there is a unique $\rho_x\in U$
such that $x$ is consistent with $\rho_x$;
\item each $\rho\in U$ has size $|\rho|\leq \UC_1(f)$.
\end{itemize}
The function $g$ will be defined on inputs $(x,y)\in\{0,1\}^{2n}$
where $x,y\in\{0,1\}^n$. If $x$ is such that $f(x)=0$, we define $g(x,y)\coloneqq *$. Otherwise if $f(x)=1$, we consider the unique $\rho_x\in U$ consistent with $x$: Denote by $r(\rho_x)\subseteq[n]$ the set
of indices $i\in[n]$ that are read by $\rho_x$. We define~$g(x,y)\coloneqq\bigoplus_{i\in r(\rho_x)}y_i$,
that is, the parity of the bits of $y$
that are indexed by $r(\rho_x)$.

To certify that $g(x,y)=b$ for $b\in\{0,1\}$, it suffices to read $\rho_x\in U$ together with the corresponding set of bits $r(\rho_x)$ in $y$. This shows that $\C(g)\le 2\UC_1(f)$. We then define the hard $*$-input by~$z\coloneqq (x,0^n)$ where $x\in f^{-1}(0)$ is any input such that $\C_0(f,x)=\C_0(f)$.
\begin{claim}
$\min\{\C_{\ol{0}}(g,z),\C_{\ol{1}}(g,z)\}\ge \C_0(f)$.
\end{claim}
\begin{proof}
Let $\rho\in\{0,1,*\}^{2n}$ be a partial input consistent with $z$ that has size $|\rho|<\C_0(f)$. Our goal is to show that $\rho$ is not a $\ol{1}$-certificate (showing that $\rho$ is not a $\ol{0}$-certificate is analogous). It is possible that~$\rho$ reads some bits in the first half of the input $z=(x,0^n)$ and some bits in the second half. We define a set~$B\coloneqq\{i\in[n]: i\in r(\rho)\text{ or }i+n\in r(\rho)\}$ that ``shifts'' all the query positions in the second half to the first half. Let $\rho'\in\{0,1\}^n$ be the partial input consistent with $x$ such that~$r(\rho')=B$. Since $|\rho'|=|B|\leq|\rho|<\C_0(f)$, we know that $\rho'$ does not certify $f(x)=0$. This means there is some $1$-certificate $\sigma\in U$ consistent with $\rho'$ and such that $r(\sigma)\not\subseteq r(\rho')$. Our goal becomes to use $\sigma$ to modify $z$ in positions outside $r(\rho)$ to obtain a $z'$ such that $g(z')=1$, which would show that $\rho$ is not a $\ol{1}$-certificate, concluding the proof. Indeed, starting with $z=(x,0^n)$ we can modify the first half $x$ to contain $\sigma$, and we can modify the bits $r(\sigma)\setminus B\neq\emptyset$ in the second half~$0^n$ so that the positions $r(\sigma)$ (in the second half) have odd parity.
\end{proof}

\subsection{Construction \ref{p2} \texorpdfstring{$\Rightarrow$}{=>} \ref{p3}}
\vspace{8pt}
\begin{mdframed}
\begin{tabular}{rl}
\bf Given: & A partial function $f\colon \{0,1\}^n \to \{0,1,*\}$ and $x\in f^{-1}(*)$.\\[6pt]
\bf Construct: & An intersecting
hypergraph $G=(V,E)$ with $|V|=2n+2$ and\\
&$r(G)=\C(f)+1$ and a colouring $c\colon V\to\{0,1\}$ such that every \\
& $c$-monochromatic hitting set has size at least $\min\{\C_{\ol{0}}(f,x),\C_{\ol{1}}(f,x)\}$.
\end{tabular}
\end{mdframed}

\noindent
For each $i\in[n]$, we introduce two vertices $v_{i,0}$ and $v_{i,1}$
into $V$. We also have two special vertices, which we denote $u_0$
and $u_1$. For each $0$-certificate $\rho\in\{0,1,*\}^n$ of size $|\rho|\leq\C(f)$,
we construct an edge $S_\rho$, as follows. For each $i\in[n]$, if $\rho_i=0$ we place $v_{i,0}$ in $S_\rho$, and if $\rho_i=1$ we place $v_{i,1}$ in $S_\rho$. We also place $u_0$ in $S_\rho$.
Then $|S_\rho|=|\rho|+1\le \C(f)+1$.

For each $1$-certificate $\rho$ of size $|\rho|\leq \C(f)$,
we construct an edge $S_\rho$ slightly differently. Essentially,
we negate the bits of $\rho$ before creating the edge out of $\rho$.
So if $\rho_i=0$ we place $v_{i,1}$ in $S_\rho$ and if $\rho_i=1$ we place
$v_{i,0}$ in $S_\rho$. We also place $u_1$ in $S_\rho$. Together,
the edges coming from $0$- and $1$-certificates constitute all the edges in $E$.
This defines $G=(V,E)$.

Note that $r(G)= \C(f)+1$. Additionally, $G$ is intersecting.
To see this, note that
if $S_\rho$ and $S_{\rho'}$ are two edges of $G$, then there are
three options: if $\rho$ and $\rho'$ are both $0$-certificates, they
share $u_0$; if $\rho$ and $\rho'$ are both $1$-certificates, they
share $u_1$; and if $\rho$ and $\rho'$ are certificates of opposite
types, then they must contradict each other at some index,
meaning that $\rho_i=0$ and $\rho'_i=1$ (or vice versa) for some $i\in[n]$.
In this last case, $S_\rho$ and $S_{\rho'}$ either both contain $v_{i,0}$
or both contain $v_{i,1}$. In all cases, $S_\rho$ and $S_{\rho'}$ have a
non-empty intersection.

We now define the colouring $c\colon V\to\{0,1\}$ based on the input $x\in f^{-1}(*)$.
We do so by setting $c(v_{i,x_i})=0$, $c(v_{i,1-x_i})=1$,
$c(u_0)=0$, and $c(u_1)=1$. It remains to prove the following claim.

\begin{claim}
If $H$ is a $c$-monochromatic hitting set for $G$, then $|H|\geq \min\{\C_{\ol{0}}(f,x),\C_{\ol{1}}(f,x)\}$.
\end{claim}
\begin{proof}
If $H$ uses the colour $1$, then it does not contain $u_0$;
since it is a hitting set, it must intersect~$S_\rho$ for each short $0$-certificate
$\rho$ in some vertex $v_{i,b}$ (where $i\in[n]$ and $b\in\{0,1\}$).
Since $H$ is monochromatic with colour $1$, we must have
$b=1-x_i$. Since $v_{i,1-x_i}\in S_\rho$, 
we must have $\rho_i=1-x_i$. In other words, the hitting set $H$
must define a set of indices in $[n]$ such that for each short $0$-certificate
$\rho$ of $f$, there is some index $i$ in this set on which $\rho$ contradicts $x$.
Since each $0$-input to~$f$ contains a short $0$-certificate
(of length at most $\C(f)$), we conclude that this set of indices
used by $H$ is such that each $0$-input to $f$ conflicts with $x$
in one of those indices. This means that we can construct a ${\ol{0}}$-certificate
by reading these indices in the string $x$; thus $|H|\ge \C_{\ol{0}}(f,x)$.

Alternatively, suppose that $H$ uses the colour $0$. Then
it does not contain $u_1$, and must intersect each $S_\rho$
for a short $1$-certificate $\rho$ of $f$ in a vertex $v_{i,b}$.
Since $H$ uses the colour $0$, we must have $b=x_i$,
and since $v_{i,x_i}\in S_\rho$, we must have $\rho_i=1-x_i$.
As before, this implies that $H$ defines a set of indices
such that each short $1$-certificate of $f$ contradicts $x$
on one of these indices; hence we can get a ${\ol{1}}$-certificate
by reading those indices in $x$, which implies
that $|H|\ge\C_{\ol{1}}(f,x)$.
\end{proof}

\subsection{Construction \ref{p3} \texorpdfstring{$\Rightarrow$}{=>} \ref{p2}}
\vspace{8pt}
\begin{mdframed}
\begin{tabular}{rl}
\bf Given: & 
An intersecting hypergraph $G=(V,E)$ and a colouring $c\colon V\to\{0,1\}$\\
&such that every $c$-monochromatic hitting set has size at least $h>r(G)$.\\[6pt]
\bf Construct: & A partial boolean function $f\colon\{0,1\}^V\to\{0,1,*\}$ and an input $x\in f^{-1}(*)$ \\
&such that $\C(f)\le r(G)$ and $\min\{\C_{\ol{0}}(f,x),\C_{\ol{1}}(f,x)\}\ge h$.
\end{tabular}
\end{mdframed}

\noindent
We define $f$ on $n=|V|$ bits so that an input to $f$ is interpreted as a colouring of $V$. We define $f(z)\coloneqq 0$ if the colouring $z$ contains a monochromatic edge of colour $0$,
and we define $f(z)\coloneqq 1$ if $z$ contains a monochromatic
edge of colour $1$. Note that both cases cannot hold, 
because $G$ is intersecting.
If neither of these cases holds, we define $f(z)\coloneqq *$. 

To certify that $f(z)=0$ or that $f(z)=1$, we can just read
a monochromatic edge in $z$; this only uses $r(G)$ bits in the worst
case over $0$- or $1$-inputs $z$, so $\C(f)\le r(G)$.

Next, consider the input $x$ to $f$ which is defined by the colouring
$c$. Since any monochromatic edge is a monochromatic hitting set
(since $G$ is intersecting, so every edge is a hitting set), and since
the minimum monochromatic hitting set in $c$ has size $h>r(G)$,
we conclude that $c$ does not have a monochromatic edge, and hence
$f(x)=*$. Observe that a certificate that $x$ is not a $0$-input
is a proof that there is no $0$-monochromatic edge in $c$, and such
a proof must necessarily read a $1$-monochromatic hitting set in $c$;
hence $\C_{\ol{0}}(f,x)\ge h$. Similarly, we have $\C_{\ol{1}}(f,x)\ge h$.

\begin{remark}
We note that $f$ is \emph{monotone} by construction: flipping any bit in an
input $z$ from $0$ to $1$ can only change $f(z)$ from $0$ to $*$ or $1$,
or from $*$ to $1$. In particular, this means that we can transform any solution to \ref{p2} into a monotone one via the steps \ref{p2}$\Rightarrow$\ref{p3}$\Rightarrow$\ref{p2}.
\end{remark}

\section{Application: Approximate degree vs.\ certificate complexity} \label{sec:apps}

Finally, we prove \autoref{cor:adeg}, which states that there exists a total function $f$ with $\C(f)\geq \tOmega(\adeg(f)^3)$. Let us quickly recall the definition of the \emph{$\epsilon$-approximate degree} $\adeg_\epsilon(f)$ of an $n$-bit boolean function~$f$: it equals the least degree of an $n$-variate polynomial $p\colon\mathbb{R}^n\to\mathbb{R}$ such that $p(x) \in f(x)\pm\epsilon$ for every boolean input $x\in\{0,1\}^n$. We also set $\adeg(f)\coloneqq\adeg_{1/3}(f)$.

\subsection{Proof of \autoref{cor:adeg}}

By applying the construction \ref{p2}$\Rightarrow$\ref{p1} (\autoref{thm:twotoone}) to our $\Hex$ function (\autoref{sec:hex}), we get a total $g$ with $\C_0(g) \geq\tOmega(\UC_1(g)^{1.5})$. All we have to show is that $g$ also has
\begin{equation}\textstyle \label{eq}
\adeg(g) ~\leq~ \tO(\sqrt{\UC_1(g)}).
\end{equation}
Let us examine the function constructed by \ref{p2}$\Rightarrow$\ref{p1} using the notation in that proof.
This proof starts out with an original $n$-bit function $f$ (namely, $\smash{\Hex_{\sqrt{n}}}$) and it defines from it a new function~$g$ on~$O(n^2 \log^2 n)$ bits using the cheat sheet framework.
An input to $g$ consists of $k\coloneqq \log n$ inputs to~$f$ and an array of size $n$, where each cell of the array is of size $k\ell$, where $\ell \leq 2 \C(f) \log n$ is the number of bits needed to specify a certificate of $f$. In a $1$-input to $g$, the correct cell, which is cell $s_z$, is supposed to contain $k$ certificates for the $k$ instances of $f$. We did not specify how the certificates would be described since the construction \ref{p2}$\Rightarrow$\ref{p1} applies to a general function $f$, but now let us describe them precisely for $f=\Hex$.
Here, a convenient $0$-certificate is a list of adjacent $0$-entries that starts from the left and ends on the right. For a $1$-certificate we can have a similar list that starts at the top and ends at the bottom. Let us modify our function $g$ to require that the certificates are presented in exactly this format.

Now for any cell $c$, consider the boolean function $g_c$ that on an input $z$ to $g$ evaluates to $1$ if $g(z)=1$ and additionally that cell $c$ is the one pointed to by $z$ (that is, $s_z = c$). 
We will show that this boolean function has an approximating polynomial of degree $\tO(\sqrt{\UC_1(g)})$. 

To check if cell $c$ is the one pointed to by the $\log n$ copies of $f$, we first need to check that the certificates contained in $c$ are valid certificates for the $\log n$ instances of $f$, and that $\log n$ $f$-outputs of these certificates, when interpreted as a number is indeed $c$.
First we claim that checking if a certificate for a particular $f$ is valid can be done with an approximating polynomial of degree~$\tO(\sqrt{\UC_1(g)})$.
Let us do this for $0$-certificates, and the construction for $1$-certificates is similar. 
Each $0$-certificate for an instance of $f$ contains $\C(f)$ many $\Hex$-matrix entries that are adjacent, all having the value $0$, and starting at the left and ending at the right.
Checking if two adjacent entries in the list correspond to adjacent matrix entries requires only $O(\log n)$ queries by a deterministic query algorithm (decision tree). There are $O(n)$ such checks to be made. Checking if a matrix entry in the list is $0$ requires $O(\log n)$ queries as well. There are $O(n)$ such checks to be made. And finally checking that the first and last entry of the list are on the left and right require $O(\log n)$ queries. In total we have to make $O(n)$ checks, each of which cost $O(\log n)$ queries. 
Equivalently, we want to compute the logical $\AND$ of $O(n)$ many query algorithms, each of which has query complexity $O(\log n)$.

A deterministic query algorithm of $O(\log n)$ queries can be converted to an exact polynomial of degree $O(\log n)$. Nisan and Szegedy~\cite{Nisan1995} showed that there is a polynomial of degree $O(\sqrt{n})$ to approximate the $n$-bit $\AND$ function. Composing this polynomial with a $O(\log n)$-degree polynomials for the individual checks gives us an approximating polynomial of degree $\tO(\sqrt{n})$ for checking if a particular certificate for $f$ is valid. Since there are $\log n$ certificates to be checked, checking all of them does not increase the degree by more than a log factor. Once we have checked if all the~$f$ certificates are valid, we know the outputs and can check if this equals $c$. Thus we have an approximating polynomial for $g_c$ of degree $\tO(\sqrt{n})$.

Now that we know that $g_c$ has an approximating polynomial of degree $\tO(\sqrt{n})$, we can construct one for $g$ from such polynomials. First we boost the approximation accuracy of the polynomials we constructed to have error $1/3n$, which only increases the degree by a log factor.
Then we observe that $g(z)=1$ if and only if one of the $g_c(z)$ functions evaluate to $1$, and furthermore, no more than one of them can evaluate to $1$ since these are unambiguous certificates. So we get an approximate polynomial for $g$ by simply summing up the polynomials for all $g_c$. Since each polynomial had error $1/3n$, the resulting polynomial has error at most $1/3$. The degree has not increased, and hence we have an approximating polynomial for $g$ of degree $\tO(\sqrt{n}) = \tO(\sqrt{\UC_1(f)})$. This proves \eqref{eq}.

\bigskip
\subsection*{Acknowledgements}

Thanks to Ryan Alweiss, Harry Buhrman, Nati Linial, and Mario Szegedy for their thoughts on \autoref{p3}. Thanks to Thomas Watson for many discussions about Hex and complexity classes.

\bigskip


\DeclareUrlCommand{\Doi}{\urlstyle{sf}}
\renewcommand{\path}[1]{\small\Doi{#1}}
\renewcommand{\url}[1]{\href{#1}{\small\Doi{#1}}}

\bibliographystyle{alphaurl}
\bibliography{hex}

\end{document}